\documentclass[a4paper,11pt]{article}
\pdfoutput=1

\usepackage[left=3cm,right=3cm,top=3cm,bottom=3cm]{geometry}

\usepackage{bbm}
\usepackage{graphicx}
\usepackage{amssymb}
\usepackage{amsthm}
\usepackage{amsmath}
\usepackage{xspace}

\theoremstyle{plain}

\newtheorem{theorem}{Theorem}[section]
\newtheorem{lemma}[theorem]{Lemma}
\newtheorem{observation}[theorem]{Observation}

\theoremstyle{definition}
\newtheorem{definition}[theorem]{Definition}

\newcommand{\ie}{i.\,e.\xspace}
\newcommand{\eg}{e.\,g.\xspace}
\newcommand{\dist}[1]{\ensuremath{\mathrm{dist}(#1)}}

\title{A Competitive Ratio Approximation Scheme\\
 for the $k$-Server Problem in Fixed Finite Metrics}
 \author{Tobias M\"{o}mke\\ Department of Computer Science\\ Saarland
 University\\ {\tt moemke@cs.uni-saarland.de}}
\begin{document}
\maketitle
\begin{abstract}
We show how to restrict the analysis of a class of online problems that includes the $k$-server
problem in finite metrics such that we only have to consider finite sequences of request. When
applying the restrictions, both the optimal offline solutions and 
the best possible deterministic or randomized online solutions only differ by at most an arbitrarily
small constant factor from the corresponding solutions without restrictions. 
Furthermore, we show how to obtain an algorithm with best possible deterministic
or randomized competitive ratio for the restricted setup.
Thus, for each fixed finite metric our result qualifies as a competitive ratio approximation scheme as
defined by G\"{u}nther et al.~\cite{GMMW13}.
\end{abstract}
\section{Introduction}
The $k$-server problem is a classical online problem that gained considerable
attention since it was proposed by Manasse et al.~\cite{MMS90}. Given an
$n$-point metric space and $k$ servers, the aim is to cover requested points
with servers such that the cost incurred by moving servers is minimized. More
precisely, the input instance is the metric space and a sequence of points that
are revealed to the online algorithm one by one. The online algorithm initially 
covers a fixed set of $k$ points with servers. After each requests it has to
adapt the set of covered points such that the requested point is included. The cost to
adapt the set of points is the sum of distances in the metric traveled by the servers.
The goal of the algorithm is to minimize the overall cost for moving servers.

One of the reasons why the $k$-server problem is important is that is is a
generalization of well known online problems such as (weighted) paging, where
there is fast memory that can store $k$ pages and the online algorithm has to decide
which of the pages is overwritten whenever a page not stored in fast memory is
requested.
At the same time, the $k$-server problem is a restriction of metrical task
systems, where a metric space of $n$ states is given and the requests
are $n$-vectors of processing costs. The online algorithm has to decide where to
move as changing states causes transition cost, but lower processing
costs may have a stronger influence than the transition.

The quality of online algorithms is most commonly measured in terms of its
competitive ratio, which is the worst case ratio of the cost of the solution computed by an
online algorithm and the cost of the best possible solution of an offline algorithm (that
knows the complete sequence of requests). To obtain bounds on the worst case cost, we
commonly construct an adversary that knows the algorithm and creates the
instance depending on the properties of the algorithm. For randomized online algorithms, it
is common to use the expected competitive ratio and to assume an oblivious
adversary. Thus, we compare the worst-case expected cost of a solution computed
by the online algorithm to the optimal cost of an offline algorithm. To create
the worst-case instance, the adversary can use the algorithm and the probability
distribution of the random bits used by the algorithm, but it does not know the
content of the random tape. The adversary has to create the complete sequence of
requests before the computation starts.

For paging, tight competitive ratios are know. There is a deterministic $k$-competitive
online algorithm \cite{ST85} and a $(H_k)-competitive$ randomized
algorithm \cite{MS91} (where $H_k$ is the $k$th harmonic number), 
whereas there are no online algorithms with better competitive ratios
\cite{ST85,FKLMSY91}.

Similarly, metrical task systems are well understood. There is a tight
deterministic competitive ratio of $2n-1$ by Borodin et al.~\cite{BLS92}. For
randomized online algorithms, the best possible expected competitive ratio is in
$\Omega(\log n/\log\log n)$ \cite{BLMN05,BBM06} and $O(\log^2n \log\log n)$
\cite{FM03}.

In contrast to these two problem and despite their close relation, the
$k$-server problem turns out to be much harder to analyze. The lower bounds on
the achievable competitive ratio are $k$ for
deterministic and an expectation of $O(\log k)$ for randomized online algorithms
\cite{MMS90,BKRS92,BLMN05,BBM06}, which matches the lower bounds for paging. In fact, these ratios are
conjectured to be tight. The best upper bound on the deterministic
competitive ratio is $2k-1$ using the so-called work-function algorithm by Koutsoupias and
Papadimitriou \cite{KP95}, but the
analysis is not known to be tight. For randomized algorithms, a recent
development of strong techniques let to an expected $O(\log^3n\log^2k \log\log n)$-competitive 
algorithm by Bansal et al.~\cite{BBMN11}. The result depends on the elegant use
of certain linear programs \cite{BBN07,BBN10a} and on embedding techniques related to
hierarchically separated trees \cite{BBN10}.
The result of Bansal et al.~\cite{BBMN11}, however, leaves plenty of space for improvements.
The most important issue is that the techniques used to obtain the algorithm inherently depend on
metric embeddings with the effect that the competitive ratio depends on the number of
points in the metric. There does not seem to be an obvious way to circumvent this
dependency without changing the direction of research.

There is an interesting approach to find algorithms with almost optimal
competitive ratio that was introduced and used for scheduling problems by G\"{u}nther 
et al.~\cite{GMMW13}. They called the concept \emph{competitive ratio approximation schemes} 
as it is somewhat related to polynomial time approximation schemes in the area of approximation
algorithms. The precise definition is that an online algorithm is a competitive
ratio approximation scheme if for any constant $\varepsilon>0$ given as a
parameter, the algorithm computes a $(c+\varepsilon)$-competitive
solution and where $c$ is the best possible strict competitive ratio of any algorithm for
that problem. At the same time it computes a value $c'$ such that $c \le  c' \le
c+\varepsilon$. The strictness was implicitly assumed when the concept was
introduced \cite{GMMW13}. We would like to emphasize that we aim for unconditional results
and thus we do not restrict the resources of online algorithms.
 
\subsection{Overview of Results and Techniques}
We show that for the $k$-server problem in fixed metrics with a finite number of
points there is a competitive ratio approximation scheme for both the deterministic
and the randomized setup. After giving some formal definitions in
Section~\ref{sec:prelim}, in Section~\ref{sec:limit} we show that accepting an arbitrarily
small error $\epsilon$ allows us to restrict our attention to finite request
sequences. We crucially use the property of the $k$-server problem to be
restartable as shown by Komm et al.~\cite{KKKM13}. However, we have to relate
the cost of an optimal solution to the number of request. The main difficulty of
the proof is to ensure that the adversary does not lose too much power when
aborting long sequences of zero-cost requests. We note that very recently, Megow
and Wiese \cite{MW13} independently obtained a competitive ratio approximation schemes
for for minimizing the makespan in the online-list model where similar to the
result of this paper they use that it is sufficient to consider a constant schedule history.

Given the restriction to finite request sequences, in Section~\ref{sec:alg} we
show how to obtain the competitive ratio approximation schemes. While for
deterministic algorithms we simply have to minimize the maximal competitive
ratio over finitely many adversaries, the randomized case is more involved. We
provide a linear program that gives a suitable sequence of probability
distributions in order to obtain a randomized competitive ratio approximation
scheme. 

Finally, in Section~\ref{sec:generalization}, we analyze the
properties of the $k$-server problem in order to determine a general class of
problems for which our techniques can be applied.

\section{Preliminaries}\label{sec:prelim}
An online-algorithm $A$ is $c$-competitive if there is a constant $\alpha$ such that
for any initial configuration and any sequence of requests, the cost of the
solution computed by $A$ is at most $c \cdot \mathrm{opt} +
\alpha$, where opt is the cost of an optimal solution obtained by an offline
algorithm. If the cost of the solution is at most $c \cdot \mathrm{opt}$, $A$ is
\emph{strictly} $c$-competitive.
Similarly, a randomized online algorithm $A$ is (strictly) $c$-competitive
against an oblivious adversary, if the cost of the solution is at most $c \cdot
\mathrm{opt} + \alpha$ (resp. $c \cdot \mathrm{opt}$) in expectation. 

We denote finite metrics used in this paper by $M=(V,d)$, where $V$ is a
finite set of points and $d$ is the corresponding function.

We use problem properties related to the definitions in Komm et
al.~\cite{KKKM13}. To \emph{reset} the computation means to forget the history
and act as if the current problem configuration is the initial configuration.
The algorithm does not know that the reset took place. A reset requires that any
configuration is initial. This can easily be achieved for the $k$-server problem
by permuting the names of the points in $V$. For a detailed analysis of such
properties we refer to Komm et al.~\cite{KKKM13}.

\begin{definition}[$D$-resetting]\label{def:resetting}
For any online problem $P$ and any positive integer $D$ we call an
algorithm $A_D$ \emph{$D$-resetting} if it resets the computation after each
multiple of $D$ paid requests (\ie, zero cost answers are not counted).
\end{definition}

We say that an online algorithm for the $k$-server problem is \emph{lazy} if it
answers requests by moving at most one server. Manasse et al.~\cite{MMS90}
showed that we may assume laziness without loss of generality.

\section{Controlling the Number of Steps}\label{sec:limit}

A metric given as input is normalized if the minimum distance between two points
is exactly one. If this is not the case, we multiply all distances by 
$\gamma := (\min_{u,v \in V}{d(u,v}))^{-1}$ to obtain a normalized version. Note
that if $M$ is finite and fixed, $\gamma$ is a constant.

\begin{lemma}\label{lem:limit}
  Let $A$ be an expected $c$-competitive online algorithm for the $k$-server problem.
  Then, for any constant $\varepsilon > 0$ and any finite metric $M$, there is a constant $D$ such 
  that there is a $D$-resetting expected $(c+\varepsilon)$-competitive online
  algorithm $A_D$ for the $k$-server problem in $M$. If $A$ is deterministic, then also
  $A_D$ is deterministic.
\end{lemma}
\begin{proof}
  As discussed above, we can assume, without loss of generality, that the minimum
  distance in the metric is one.
  Let us first assume that $A$ is a deterministic algorithm in order to clarify the 
  line of argumentation. We assume that $c$ is at most a constant depending only on $k$ 
  (otherwise we use the $2k-1$-competitive algorithm).

  We first give a lower bound on the cost of an optimal solution such
  that the competitive ratio of the claim is satisfied. The analysis is analogous 
  to that by Komm et al.~\cite{KKKM13}.
  Let $B := k \cdot  \max_{s,t \in V} d(s,t)$, which is an upper bound on the maximum
  transition cost within the metric space $M$. Then restarting an algorithm causes an 
  additional cost of at most $B$ compared to an optimal algorithm (which may
  have arranged the positions of the $k$ servers as to optimize the transition
  costs for the subsequent answers).  To be $(c+\varepsilon)$-competitive, $A_D$
  has to ensure that the optimal cost opt for the given request sequence is high
  enough such that $\varepsilon \cdot \mathrm{opt} \ge B$. Thus, if the restarts
  are performed after an optimal cost of $\mathrm{opt} \ge B \varepsilon^{-1}$
  is reached, the claim is true. We call the computations between resets the
  \emph{phases} of the computation. Note that $A$ is $c$-competitive in each of the
  phases that are created due to the restarts. Thus it is sufficient to consider
  a single phase such that the analysis is independent of the initial
  configuration of servers and we aim to find an upper bound $D$ on the minimum number of
  requests per phase such that the adversary keeps its full strength.

  In general, the value of opt is independent of the number of request as an optimal
  solution may have arbitrarily long sequences of zero cost answers. We now
  argue that for each instance, the lengths of sequences that can be
  answered with zero cost by an optimal offline solution beyond a certain bound cannot
  help the adversary.

  Let us fix an optimal solution Opt to the offline problem for the
  sequence of points requested by the adversary. We assume without loss of
  generality that each time Opt has a zero-cost answer, $A$ moves.
  Otherwise the adversary may simply skip that request. The skipping may cause
  $A$ to change its internal state and thus its behavior, but then there is another
  $c$-competitive algorithm that behaves just like $A$ except that it does not change its
  behavior due to a zero cost answer. This cannot lead to a competitive ratio
  worse than $c$ as the adversary can exploit differing behavior without
  changing opt. Therefore, we assume that $A$ ignores any request
  that does not require a movement of servers.

  Let $\phi := c\cdot opt + \alpha$, where $\alpha$ originates from the definition of competitiveness
  and is a constant depending on $A$.  
  In each phase, if Opt has a zero-cost answer more than $\phi$ times,
  $A$ has moved too much to be $c$-competitive as due to the normalization each move
  has a cost of at least one. Therefore, the total number of requests per phase is limited to
  $D=\phi \cdot \mathrm{opt}$, \ie, $\phi$ requests for each step in which
  opt moves with non-zero costs.

  Now we discuss the differences if $A$ is a randomized algorithm with
  expected competitive ratio $c$. As before, if the optimal cost of each phase
  is at least opt, we are done. However, it is more complicated to bound the
  number of zero-cost answers of an optimal solution. Let us fix an arbitrary
  adversary for $A$ and an optimal offline solution Opt to the sequence of
  the adversary's requests. Let $X_i$ denote the expected value of the cost of
  the solution computed by $A$ for the $i$th request (with respect to the whole
  phase, \ie, not depending on the previous answers).
  
  Similar to the deterministic case we can assume that $X_i$ is
  nonzero if Opt has a zero cost answer to the $i$th request, but the value of
  $X_i$ may be much smaller than one.

  We now argue that the adversary may skip certain 
  steps without losing too much. Let us double the bound that we used in order to
  specify opt and assume that the cost of an optimal offline solution is at least 
  $\mathrm{opt} \ge 2 B \varepsilon^{-1}$. (If opt is smaller, then there is no
  subsequent reset and thus we do not have to consider the allowance of
  $2\varepsilon B$.)
  This way, performing a reset each time the optimal offline solution has a cost of at least
  opt leads to an expected competitive ratio of at most $c + \varepsilon/2$ and
  we gain the freedom to accept an additional loss of $\varepsilon \cdot \mathrm{opt}/2$ in
  each phase.

  We divide each phase into at most opt sub-phases, each of them (except the
  first one) starting with a request
  that causes Opt to answer with non-zero cost (\ie, a cost of at least one). 
  We show that we may modify $A$ such that, after constantly many requests in
  each sub-phase, the adversary can skip the remaining zero-cost requests
  without losing power and the modified algorithm has an additional expected
  cost of at most $\varepsilon/2$. Thus, after the $i$th sub-phase the total
  increase of the expected costs due to the modification is at most 
  $i \cdot \varepsilon/2$.

  We analyze the randomized online algorithm as an infinite collection of
  deterministic algorithms, one for each string of random bits. Furthermore we
  assume without loss of generality that each of these algorithms is lazy.
 
  We use that any of these algorithm is forced to position a server at
  the requested point. Thus, if $k=1$, similar to the deterministic case we can
  assume that the length of a phase is exactly one as repeated requests of the
  same point cause zero expected cost of the algorithm.
  To simplify the presentation, we first analyze the case of exactly
  two servers and generalize to arbitrary constants $k$ afterwards.

  An infinite zero-cost request sequence of the adversary has to alternate
  between two points $p$ and $q$ occupied by the two servers. After the first
  request of the sequence (say $p$), each of the algorithms has to move servers
  in each step until it reaches the same configuration as the adversary (namely
  one server on $p$ and the other one on $q$). We call this the \emph{correct}
  configuration. Thus, for each of the algorithms the generated cost is one per
  request until the correct configuration is reached.
  Afterwards all remaining requests cause zero cost. After 
  $2(c \cdot opt + \alpha)$ requests, the probability to have chosen an algorithm with
  the correct configuration is at least $1/2$ as otherwise the expected
  competitive ratio is larger than $c$. Similarly, after $2B(c \cdot opt + \alpha)/\varepsilon$
  requests, the probability to have selected an algorithm that did not yet enter
  the correct configuration is at most $\varepsilon/(2B)$. We change $A$ such
  that if the adversary alternately requests two points for 
  $\xi_2=2 + 2B(c \cdot opt + \alpha)/\varepsilon$ steps, all corresponding algorithms
  deterministically enters the correct configuration. 
  (We added two steps to allow Opt to enter the correct configuration.) The configuration change
  costs at most $B$, which leads to an expected cost of $\varepsilon/2$. The
  remaining discussion of $k=2$ is analogous to the deterministic case, except
  that we use the parameter $\varepsilon/2$ such that the total loss of the
  competitive ratio adds up to $\varepsilon$. (The loss of $\mathrm{opt} \varepsilon/2$
  due to the modification of the algorithm corresponds to a loss of $\varepsilon/2$ for
  the competitive ratio.)

  To generalize the argumentation to arbitrary $k$, we determine a factor on the
  number of requests recursively. For any $i$, the variable $\xi_i$ determines the number of
  required requests. Suppose there are $i+1$ servers and the adversary uses at
  most $i+1$ points in the sub-phase. Let $p$ one of the $i+1$ points such that
  within the sub-phase, a maximal subsequence of requests to only $i$ points does not
  include $p$. Then, the length of the subsequence is at most $\xi_i$ as otherwise a
  configuration change to cover all $i+1$ points would be affordable.
  Therefore, after $\xi_i \cdot \xi_2$ requests, each algorithm that did not enter the
  correct configuration has at least $\xi_2$ transitions to not covered servers
  and setting $\xi_{i+1} = \xi_i \cdot \xi_2$ is sufficient. As $k$ is a constant, also $\xi_k$ is.
\end{proof}

The following theorem follows directly from Lemma~\ref{lem:limit}.

\begin{theorem}\label{cor:limited}
  Let $A$ be an online algorithm for the $k$-server problem such that its
  expected competitive ratio is $c$. Then for any $\varepsilon > 0$ and any
  fixed finite metric there is a constant size adversary (depending on $\alpha$) such that on
  the instance of the adversary, the expected competitive ratio of $A$ is at least $c - \varepsilon$.
\end{theorem}

\section{Algorithms}\label{sec:alg}
Due to Theorem~\ref{cor:limited}, there is a constant number of different request
sequences that we have to consider for any online algorithm in order to determine its
competitive ratio up to an arbitrarily small error. This enables us obtain a
deterministic online algorithm with almost best-possible competitive ratio by finding a strategy
that minimizes the maximum cost
solution over all adversaries. Let $m$ be the maximum
number of requests to be considered. Then there are $n^m$ different request
sequences and $k^m$ different answer sequences if we assume the algorithm to be
lazy. Thus the number of different algorithmic strategies is the number of mappings from the set
of requests to the set of answers, \ie, $(k^m)^{(n^m)}$. 

For randomized algorithms the situation is more involved, because the answers
are not single servers but probability distributions over the servers. Therefore
the number of different randomized online algorithms is infinite. In order to
also handle randomized algorithms, we determine a rational convex polyhedron $P$
depending on a threshold parameter $\tau$ such that, if $P$ is non-empty, there
is a randomized algorithm such that $\tau$ is its strict expected competitive ratio. We
determine the value of $\tau$ (up to an arbitrarily small error) by binary search.
To determine $P$, we will use constantly many linear constraints and thus one can determine
whether the polyhedron is empty.

We denote a single request by $r \in \{1,\dotsc,n\}$ and a sequence of $t$
requests by $\rho_t = \{r_i\}_{i=1}^{t}$. A configuration $C$ is a subset of $k$
points in the metric, \ie, $C \subseteq 2^{\{1,\dotsc n\}}, |C| = k$. 
The answer to a request $r$ is a configuration $C$ that satisfies the request.
The sequence of the first $t$ answers is denoted by $\sigma_t =
\{C_i\}_{i=1}^t$. To simplify notation we implicitly assume that the names $r_i$
and $C_i$ belong to the sequences $\rho_i$ and $\sigma_i$ and that, \eg, $r_i'$ 
and $C_i'$ belong to $\sigma_i'$ and $\rho_i'$. We sometimes concatenate
sequences and elements. For instance $\sigma_i C$ appends a single configuration
$C$ to the sequence $\sigma_i$.
Let $T$ be the total number of request that we have to consider, obtained by
applying Theorem~\ref{cor:limited}.
We define $S_t(\rho_t)$ to be the set of all feasible answer sequences on a given
request sequence $\rho_t$. To simplify notation, we write $\sigma_t(\rho_t)$ instead of
$\sigma_t \in S_t(\rho_t)$.

Then, for $1 \le t \le T$, our linear program has variables 
$x(\rho_t,\sigma_t)$ for all possible
sequences $\rho_t$ and all answer sequences $\sigma_{t}$ compatible with
$\rho_t$, that is, with $\sigma_t \in S(\rho_t)$.

The configuration $C_0$ is the start configuration (or 
end configuration of the previous phase). 
The meaning of the value of the variable is 
$x(\rho_t,\sigma_t) = p(\sigma_t \mbox{ is selected} \mid \mbox{the adversary
requested $\rho_t$})$.  

Let $\dist{C,C'}$ be the minimum cost to move the servers from
configuration $C$ to $C'$. For a sequence of configurations
$\{C_i\}_{i=0}^k$, $\dist{\sigma} := \sum_{i=1}^{k} \dist{C_{i-1},C_i}$.

\begin{align}
  \label{eq:prob}
  \sum_{\sigma_t(\rho_t)} x(\rho_t,\sigma_{t}) &= 1& \mbox{for all }t, \rho_t\\
  \label{eq:consistent}
  \sum_{C(r_t)} x(\rho_t,\sigma_{t-1}C) 
  &=  x(\rho_{t-1},\sigma_{t-1})
  &\mbox{for all }t>1, \rho_t, \sigma_{t-1}(\rho_{t-1})\\
  \label{eq:comp}
  \sum_{\sigma'_T(\rho_T)} x(\rho_T,\sigma'_T) \dist{\sigma'_T}& 
  \le \tau \dist{\sigma_T}&
  \mbox{for all }\rho_T, \sigma_T(\rho_T)\\
  \label{eq:nonzero}
  x(\rho_t,\sigma_{t}) &\ge 0& \mbox{for all }t, \rho_t, \sigma_{t}
\end{align}

The first set of constraint \eqref{eq:prob} together with \eqref{eq:nonzero} ensures
that we obtain probability distributions over all answer strings for each $t$,
$\rho_t$.
Furthermore, \eqref{eq:consistent} ensures that the distributions in consecutive
time steps are consistent.
The set of constraints \eqref{eq:comp} ensures
that the expected competitive ratio is smaller than $\tau$.

We may simplify some of the constraints and obtain
\begin{align}
  -\sum_{C(r)} x(r,C) &\le  - 1& \mbox{for all requests }r\\
  x(\rho_{t-1},\sigma_{t-1}) - \sum_{C(r_t)}  x(\rho_t,\sigma_{t-1}C)&
  \le 0&
  \mbox{for all }t>1, \rho_t, \sigma_{t-1}(\rho_{t-1})\\
  \sum_{\sigma'_T(\rho_T)} x(\rho_T,\sigma'_{T}) \dist{\sigma'_T} 
  &\le \tau \dist{\sigma_T}&
  \mbox{for all }\rho_T, \sigma_T(\rho_T)\\
  x(\rho_t,\sigma_{t}) &\ge 0& \mbox{for all }t, \rho_t, \sigma_{t}
\end{align}

Given a point in the polyhedron, the corresponding algorithm determines the
probability distribution $x'$ such that
\begin{align*}
x'(\rho_t,\sigma_t-1 C) =\,& p(\mbox{$C$ is selected} \mid \mbox{The adversary
requested $\rho_t$ }\\
& \mbox{and the previous sequence of answers was $\sigma_{t-1}$}).
\end{align*}
This is exactly
\[
x'(\rho_t,\sigma_{t-1} C') = 
\frac{x(\rho_t,\sigma_{t-1} C')}{\sum_{C: r_t \in C}x(\rho_t,\sigma_{t-1} C)}.
\]

\section{Generalization}\label{sec:generalization}

In the paper we focused on the $k$-server problem as it is the most important
problem in the context of our research. The used techniques, however, can be
generalized.

The properties of the $k$-server problem that we used in the proof of
Lemma~\ref{lem:limit} boil down to properties such as that we have to be able to
normalize the step costs, restart the computation at any configuration, and there is an upper
bound on the cost in each step. To rigorously specify these properties,
we can use the notation from Komm et al.~\cite{KKKM13}. To give a short summary of the
terminology, a partition function of an online problem is an assignment of costs
to separate steps of the computation. For a given partition function, a state
is an equivalence class of configurations that are not distinguishable with
respect to the cost. A problem is symmetric if the computation can begin in each
state and it is opt-bounded if, for any pair of state, the optimal offline cost of any
sequence of requests only differs by a constant when starting from either of the
states. A problem is request-bounded if the maximal cost to answer a single
request has a constant upper bound. For our results, we need one additional property.
\begin{definition}[Step Cost Normalization]
  A partitionable online problem $P$  with the goal to minimize costs is called
  \emph{step-cost normalizable} if there is a constant $\gamma$
  such that the minimum cost of any non-zero cost answer to a request is at least
  $\gamma^{-1}$. The scaled problem where all costs are multiplied by $\gamma$
  is the \emph{step-cost-normalized version of $P$} and the $\gamma$ is called the
  normalization factor.
\end{definition}
Using these properties, we obtain the following result.
\begin{observation}\label{obs:gen}
Lemma~\ref{lem:limit} can be adapted to any symmetric,
opt-bounded, request-bounded, step-cost normalizable minimization problems that
allows for an algorithm with (expected) constant competitive ratio and that has finitely many
states such that the transition between states has nonzero cost and all states
are reachable from any state. We can use algorithms analogous to those in Section~\ref{sec:alg}.
\end{observation}
\begin{proof}[Proof sketch]
The opt-boundedness together with the symmetry property enables resets of the
algorithm such that there is only a constant loss. The step-cost normalizability
is necessary to be able to fix a constant number of steps where the optimal
solution has nonzero costs. The request-boundedness enables the forced configuration
changes in the proof for randomized algorithms.

Instead of using the laziness assumption for $k$-server algorithms, we can bound
the number of steps using a variable $\xi'_n$ that is analogous to $\xi_k$ but
uses the frequency of any state (which requires that the state transitions have
nonzero costs and that all states are reachable).
\end{proof}

Note that the properties required by Observation~\ref{obs:gen} are fulfilled by
task systems if we fix a threshold such that the processing time of each state
(in the notation of task systems, not the equivalence classes defined above) in
each request is either zero or larger than the threshold.

The use of general properties in Observation \ref{obs:gen} is useful as this way
it can be readily applied to restricted problems. Note that a concrete problem
cannot offer that strength because restrictions of the problem may prevent
modifications of the algorithm's behavior as they are done in the proof.

\section*{Acknowledgment}
I would like to thank Dennis Komm for helpful discussions that inspired some of
the ideas. This work was supported by Deutsche Forschungsgemeinschaft grant BL511/10-1.

\end{document}